%% file: main.tex
\documentclass[10pt,conference]{IEEEtran}

\usepackage[ruled]{algorithm2e} 

\usepackage{listings}
\usepackage{graphicx,amssymb,mathtools,amsthm}
\usepackage{stmaryrd}
\usepackage{longtable}
\usepackage{float}
\usepackage{url}
\usepackage{enumerate}
\usepackage{wrapfig}
\usepackage{subcaption}
\usepackage{pifont}
\usepackage{adjustbox}
\usepackage{tikz}
\usepackage[colorinlistoftodos]{todonotes}
\usepackage[colorlinks=true, allcolors=blue]{hyperref}

\newtheorem{thm}{Theorem}
\newtheorem{prop}{Proposition}

\newtheorem{defn}{Definition}

\title{An Axiomatic Approach to Detect Information Leaks in Concurrent Programs} 

\author{\IEEEauthorblockN{Sandip Ghosal}
\IEEEauthorblockA{\textit{Department of Computer Science \& Engineering} \\
\textit{Indian Institute of Technology, Bombay}\\
Mumbai, India \\
sandipsmit@gmail.com}

\and

\IEEEauthorblockN{R. K. Shyamasundar}
\IEEEauthorblockA{\textit{Department of Computer Science \& Engineering} \\
\textit{Indian Institute of Technology, Bombay}\\
Mumbai, India \\
shyamasundar@gmail.com}
}

\begin{document}


\maketitle

\begin{abstract}
Realizing flow security in a concurrent environment is extremely challenging, primarily due to non-deterministic nature of execution. The difficulty is further exacerbated from a security angle if sequential threads disclose control locations through publicly observable statements like \textit{print}, \textit{sleep}, \textit{delay}, etc. Such observations lead to \textit{internal} and \textit{external} timing attacks. Inspired by previous works that use classical Hoare style proof systems for establishing correctness of distributed (real-time) programs, in this paper, we describe a method for finding \textit {information leaks} in concurrent programs through the introduction of  \textit{leaky assertions} at observable program points. Specifying leaky assertions akin to classic assertions, we demonstrate how information leaks can be detected in a concurrent context. To our knowledge, this is the first such work that enables integration of different notions of {\it non-interference} used in functional and security context. While the approach is sound and relatively complete in the classic sense, it enables the use of algorithmic techniques that enable programmers to come up with leaky assertions that enable checking for information leaks in sensitive applications.
\end{abstract} 

\input{introduction}
\input{language_syntax}
\input{background}

\input{timing_sensitive}
\input{related_work}
\input{conclusions}
\bibliographystyle{IEEEtran}
\bibliography{main}

\end{document}

%% file: introduction.tex
\section{Introduction}
\label{sec:introduction}
Information flow control (IFC) in sequential and concurrent programs has been an active area of research for over two decades. While researchers have developed specification approaches  \cite{zdancewic2003observational,boudol2002noninterference,volpano_smith,goguen1982security} and enforcement mechanisms for sequential programs \cite{buiras2015hlio,stefan2011flexible,zheng2007dynamic,hammer2009flow,simonet2003flow,myers2001jif,denning1977certification}, flow security in concurrent programs remains a challenge primarily due to covert channels that could leak secret information.  {\it Timing channels} in concurrent programs are the most commonly known covert channels.

\begin{figure}[htb]
\centering
	\begin{minipage}[b]{3.8cm}
		\begin{lstlisting}[frame=single]
await sem>0 then
  sem = sem - 1;
  print(`a');
  v = v+1;
  print(`b');
  sem = sem + 1;
		\end{lstlisting}
	\end{minipage}
	\hspace{1mm}
	\begin{minipage}[b]{4.2cm}
		\begin{lstlisting}[frame=single]
print(`c');
if h then
  await sem>0 then
    sem = sem - 1;
    v = v+2;
    sem = sem + 1;
else
  skip;
print(`d');
		\end{lstlisting}
	\end{minipage}
\caption{(a) Process $T_1$ (left); (b) Process $T_2$ (right)}
\label{fig:observable_leak}
\end{figure}
For instance, often programs encode secret information in such a way that allow an attacker to learn about confidential data by observing time variations between interleaving execution of assignments to public variables or output statements such as \textit{print}, or the availability of shared resources in a particular time interval. 
Recent flow security enforcing mechanisms for sequential programs \cite{buiras2014dynamic,stefan2011flexible}  identify information leaks in individual threads, but the approaches are too restrictive for programs that reveal the control position through the observable output statements. For example, consider the program shown in Figure \ref{fig:observable_leak}, (Cf. \cite{le2007automaton}) where $h$ is a secret variable, and \texttt{await} command is a conditional critical region. The \texttt{await} command essentially represents a number of standard synchronization primitives including semaphores. Let us suppose that process $T_1$ holds the semaphore and executes the critical region as an indivisible operation. Then,  output `acdb' arising due to concurrent executions of $T_1$ and $T_2$ reveal the value of $h$ as 0. There have been a  few efforts aiming to prevent timing leaks in concurrent or multi-threaded programs \cite{vassena2019foundations,vassena2017securing,popescu2013formal,barthe2010security,russo2006closing,smith2001new,agat2000transforming,smith1998secure}  but majorly have sidestepped the issue that arise due to output statements.

In this paper, we aim to preserve the confidentiality of programs in a concurrent context preventing an attacker external to the system from gaining knowledge about secret information through observations of program outputs. We describe an axiomatic approach that establishes functional correctness \cite{owicki1976axiomatic} in a parallel context to realize the insecurity at a program point that might leak information due to concurrent executions. In the envisaged work, we introduce auxiliary time variables that capture execution time for each execution step and enables us to derive a notion of leaky assertion for each public program locations, that when satisfied, one can conclude that the program leaks out information. Enforcing functional non-interference in a concurrent setting that uses leaky assertions leads to the  identification of plausible information leaks at various program locations. Our proposed approach is primarily motivated by approaches used by Lamport and Schneider \cite{schneider1991putting,lamport1982assertional,lamport2005real,abadi1994old} on the correctness of real-time distributed programs. Similar to these approaches using  time as an ordinary variable for specifying and reasoning about real-time systems, we we arrive at leaky assertions using an external time variable. With such a specification leaky and standard assertions for concurrent programs, we check for non-interference using  verification conditions similar to those used in `Gries-Owicki' \cite{owicki1976axiomatic}. By showing the program is non-interference, we infer information leaks through the leaky assertions embedded in the specification.

Main contributions of the paper are:
\begin{enumerate}
\item Introduction of leaky assertions when embedded in concurrent program specifications, enables detection of information leaks once the proof of components of the concurrent programs are shown to be correct when considered in isolation and are further shown to be interference free.
\item Demonstration of the method of embedding leaky assertions  in standard concurrent program specification  \cite{owicki1976axiomatic}, can be used to assert information leaks, from which the information leaks are detected.
\item Checking for information leaks due to public observers like IO and time sensitive operations (those that could disclose program locations through output statements or observable timing channels), in non-interferring  flow secure concurrent programs \cite{andrews1980axiomatic}.
\end{enumerate}

Rest of the paper is organized as follows. Section \ref{sec:language_syntax} provides an abstract syntax of the concurrent language. Section \ref{sec:background} briefly describes  the necessary background on correctness of concurrent programs and  non-interference.  Section \ref{sec:timing_sensitive} illustrates our solution followed by a formal description of our approach and proof of soundness in Section \ref{sec:axiomatic}. Section \ref{sec:related_work} compares our approach with other related work in the literature, followed by conclusions in Section \ref{sec:conclusions}.

%% file: language_syntax.tex
\section{Language Syntax}
\label{sec:language_syntax}
Before proceeding with the solution, we shall briefly describe the syntax of the language we use.

Let $s$ ranging over  set $S$ be the set of commands, $x$ be the set of variables ranging over $X$, and $e$ denote an expression. Language primarily comprises basic actions such as \texttt{skip}, assignment, sequence and control statements such as conditional branch, iteration, \texttt{await} as in \cite{owicki1976axiomatic}. In addition, we consider two  statements  \texttt{print(e)} and \texttt{delay($e$)} that serve as observable points from the perspective of security and hence, these statements are often termed {\it public\/} statements. Statement \texttt{print($e$)} writes $\llbracket e\rrbracket$, where $\llbracket . \rrbracket$ denotes the classic evaluation of expressions, into the standard output file that is  visible to public. On the other hand, \texttt{delay($e$)} postpones the current execution for  $\llbracket e\rrbracket$ time units. The abstract syntax of the language is given below:
\begin{tabbing}
$s$ ::=\=\hspace{0.5cm}{\bf \texttt{skip}} $|$ $x=e$ $|$ $s_1;s_2$ $|$ {\bf \texttt{if}} $e$ \textbf{\texttt{then}} $s_1$ {\bf \texttt{else}} $s_2$ \\
\>\hspace{0.5cm}$|$ {\bf \texttt{while}} $e$ \textbf{\texttt{do}} $s$ \textbf{\texttt{done}} $|$ \texttt{print}($e$) $|$ \texttt{delay}($e$) \\
\>\hspace{0.5cm}$|$ $s_1\parallel s_2$ $|$ \texttt{await} $b$ \texttt{then} $s$\\
\end{tabbing}

\noindent Operator $s_1\parallel s_2$ denotes a parallel (or concurrent)  composition of sequential programs $s_1$ and $s_2$. The \textbf{\texttt{await}} command is a \textit{conditional critical region}, that turns a statement $s$ into an indivisible action, is executed in a state that satisfies condition $b$. The semantics of sequential part of our language essentially follows the classical while-language. Whereas, the semantics for concurrent composition and \texttt{await} statements are in line with the transition rules given in \cite{brookes1996full}. 

%% file: background.tex
\section{Correctness of Concurrent Programs}
\label{sec:background}


The seminal work of Hoare \cite{hoare1969axiomatic} provided a basis for establishing axiomatic correctness of imperative programs. Let $P$ and $Q$ be assertions referred to as pre-condition and post-condition for program/statement $S$. Hoare's triple specification for partial correctness of program $S$, is denoted $\{P\}S\{Q\}$ with the following interpretation:  if \textit{pre-condition} $P$ is true before execution of $S$, then \textit{post-condition} $Q$ holds if and when the execution of $S$ terminates.  

Owicki and Gries \cite{owicki1976axiomatic} generalized  Hoare's axiomatic approach  for shared variable concurrency with the introduction of   non-interference (functional) of programs.
Let $S_i$ for i =1,...n, be sequential programs.  $S_1\parallel S_2 \parallel\dots\parallel S_n$ denotes concurrent composition.
The effect of concurrent execution is the same as executing each  $S_i$ as an independent program if the programs are disjoint. If the programs are not disjoint, then, it is necessary that programs be {\it interference free}. Non-interference requires that assertions used in the proof $\{P_i\}S_i\{Q_i\}$ for each process remain validated under concurrent execution of the other processes. Let, $T$ be a statement in program $S_i$. Then, concurrent execution of  $\{P_1\}S_1\{Q_1\},\dots,\{P_n\}S_n\{Q_n\}$ is said to be \textit{non-interfering} if $T$ does not interfere with $\{P_j\}S_j\{Q_j\}$ where $i\neq j$. 
Formal definition is given below.

\begin{defn}[\textbf{Correctness}]
\label{defn:functional_correctness}
Given a proof $\{P\}S\{Q\}$ and a statement $T$ with the pre-condition $pre(T)$, the statement $T$ is said to be \textit{non-interfering} with $\{P\}S\{Q\}$ if the following two conditions hold:
\begin{enumerate}
	\item $\{Q\wedge pre(T)\}T\{Q\}$;
	\item $\{pre(S')\wedge pre(T)\}T\{pre(S')\}$ where, $S'$ is an \texttt{await} or assignment statement within $S$ but not within an \texttt{await} block.
\end{enumerate} 
\end{defn}

%% file: timing_sensitive.tex
\section{Capturing Information Leaks Through Output}
\label{sec:timing_sensitive}

The crux of our approach lies in setting observers from which one can draw inference about information leaks. Observation points allow us to deduce how long a process could wait at a control point relative to another process at some other control point wherein the two processes share information via shared variables (could be sensitive) through synchronization. This is realized through auxiliary variables that are used in the specification of concurrent programs. Lamport \cite{lamport2005real} illustrates an effective use of auxiliary variables for establishing correctness of real-time programs. 



Timing constraints for the waiting time that impose upper or lower bound for an action to occur can be expressed using auxiliary (time) variables. Semantic characterization of time constraints required for mutual exclusion and analysis of bounded waiting time is similar to that shown in \cite{lamport2005real}.



\subsection{ Illustrative Example}
\label{subsec:solution}

Let us illustrate that the labelled program $T_2$ showin in Figure \ref{fig:location-leaks}, where labels denote identifiers for statements as well as locations of program control does indeed leak information via public statements. Labels $l_0$ and $l_8$ are the entry and exit labels respectively. Let us mark the execution time of the program at specific program locations using auxiliary variable $t$.
\begin{figure}[htb]
\centering
        \begin{lstlisting}[numbers=none,framexleftmargin=0em]
$\color{blue}l_0:$  $\mathbf{\langle\langle t=T\rangle\rangle}$ print(`c');
$\color{blue}l_1:$ if h then
  $\color{blue}l_2:$ await sem>0 then
    $\color{blue}l_3:$ sem = sem - 1;
    $\color{blue}l_4:$ v = v+2;
    $\color{blue}l_5:$ sem = sem + 1
   else
    $\color{blue}l_6:$ skip;
$\color{blue}l_7:$ $\mathbf{\langle\langle t=t'\rangle\rangle}$ print(`d');
$\color{blue}l_8:$
        \end{lstlisting}
\caption{Process $T2$ (Cf. Figure \ref{fig:observable_leak}) with auxiliary variable $t$}
\label{fig:location-leaks}
\end{figure}
Execution of statements $l_3$ to $l_5$ is subject to availability of  $sem$, and  hence, involves a delay, say $\Delta t$. Now, an observer can easily measure the intermediate time difference between execution of public statements $l_0$ and $l_7$ and just a comparison with $\Delta t$ would enable him to deduce that information in $h$ more deterministic (informally, cardinality of the set of possible values gets reduced from the original set of possibilities)  than he could infer otherwise.

\begin{align*}
\footnotesize
    \begin{array}{c}
\dfrac{[at\_l_0]:t=T\quad[at\_l_7]:t=t'\quad t'-T\ll\Delta t}{h=0}\\
\\
\dfrac{[at\_l_0]:t=T\quad[at\_l_7]:t=t'\quad t'-T\geqslant\Delta t}{h=1} \\
	\end{array}
\end{align*}

In short, our approach is based on specifying postulates of leaky assertions at expected control points in the program and establish  correctness of the program in a standard manner. Once the proof is established, the leaky assertions locations indicate the points of leakage of information; the actual path can be derived using the semantic rules of the language for the program. 

Thus, embedding {\it leaky assertions} in the specification of functional correctness (Definition \ref{defn:functional_correctness}) of programs leads to a  proof system  that enables asserting leaks in the programs . Let $S$ to be an assignment or a non-nested \texttt{await} statement from the thread $T_1$. Let $P$ and $Q$ be pre- and post-conditions of $S$ respectively. Then,  proof of $\{P\}S\{Q\}$ asserts that the  leaky assertions hold when $T_1$ and $T_2$ execute concurrently.
\begin{enumerate}
	\item $\{Q\hspace{3mm}\wedge\hspace{3mm}((([at\_l_0]:t=T\wedge[at\_l_7]:t=t'\wedge t'-T\ll\Delta t)\implies (h=0))\hspace{3mm}\vee\hspace{3mm}(([at\_l_0]:t=T\wedge[at\_l_7]:t=t'\wedge t'-T\geqslant\Delta t)\implies (h=1)))\}$\\\texttt{print(`d')}\\$\{Q\}$\\
	\item $\{P\hspace{3mm}\wedge\hspace{3mm}((([at\_l_0]:t=T\wedge[at\_l_7]:t=t'\wedge t'-T\ll\Delta t)\implies (h=0))\hspace{3mm}\vee\hspace{3mm}(([at\_l_0]:t=T\wedge[at\_l_7]:t=t'\wedge t'-T\geqslant\Delta t)\implies (h=1)))\}$\\\texttt{print(`d')}\\$\{P\}$\\
	\item $\{((([at\_l_0]:t=T\wedge[at\_l_7]:t=t'\wedge t'-T\ll\Delta t)\implies (h=0))\hspace{3mm}\vee\hspace{3mm}(([at\_l_0]:t=T\wedge[at\_l_7]:t=t'\wedge t'-T\geqslant\Delta t)\implies (h=1)))\hspace{3mm}\wedge\hspace{3mm}P\}$\\S\\$\{((([at\_l_0]:t=T\wedge[at\_l_7]:t=t'\wedge t'-T\ll\Delta t)\implies (h=0))\hspace{3mm}\vee\hspace{3mm}(([at\_l_0]:t=T\wedge[at\_l_7]:t=t'\wedge t'-T\geqslant\Delta t)\implies (h=1)))\}$
\end{enumerate}

Note that the leaky assertions successfully capture the execution time inconsistencies for both the cases where semaphore $sem$ is locked or released by the process $T_1$. Therefore, the assertions remain unaffected by the pre-and post-conditions for any statement $S$ in $T_1$. As the assertions always hold good, this invariably proves the program as leaky.

\noindent{\textbf{Remarks:}} It is interesting to understand the impact if the program contains a \texttt{delay} statement; suppose the statement \texttt{skip} is replaced by {\it delay(50)} -- a delay of 50ms (milliseconds for instance). Now, consider the following cases with respect to the leaky assertions:
\begin{enumerate}
    \item $\Delta t\ll 50ms:$\\ \begin{align*}\dfrac{[at\_l_0]:t=T\quad[at\_l_7]:t=t'\quad \Delta t \ll (t'-T)\geq 50ms}{h=0}\\ \dfrac{[at\_l_0]:t=T\quad[at\_l_7]:t=t'\quad \Delta t \leq (t'-T)\ll 50ms}{h=1}\end{align*}
    \item $\Delta t\gg 50ms:$\\
    \begin{align*}\dfrac{[at\_l_0]:t=T\quad[at\_l_7]:t=t'\quad \Delta t \gg (t'-T)\geq 50ms}{h=0}\\ \dfrac{[at\_l_0]:t=T\quad[at\_l_7]:t=t'\quad \Delta t \leq (t'-T)\gg 50ms}{h=1}\end{align*}
    \item $\Delta t\approx 50ms:$\\ 
    It would be difficult to postulate leaky assertion as an observer cannot easily determine the execution path for $h=0$ or $h=1$; hence, the program may be considered secure even though it may not be; note that axiomatic proof is relative to the specifications.
\end{enumerate}
\section{Formalization of our Approach}
\label{sec:axiomatic}
In our two-language approach for establishing leaks, if we can show that assertions at various program points imply information leaks, then we would have established that the program is not-secure; in other words, it is leaky. Thus, if we establish that the program is insecure, we would have established a counter-example; on the other hand, if we establish that the program is secure, it would imply that the program is secure relative to the given specification. Below, we establish the  first from which the second follows naturally. 



\begin{defn}[\textbf{Leaky Assertion}]
\label{defn:leaky_assertion}
For a given set of assertions $\langle a_0,\dots,a_n\rangle$ at respective control points $\langle \ell_0,\dots,\ell_n\rangle$ of the given program $S$, 
using a nonempty set  of sensitive (or high) variables $h_1,\dots,h_m$ of the program, is said to be leaky if $a_i$ determinizes the value of at least one of its' sensitive variables, say $h_j$.
\end{defn}

As discussed earlier, we say that an execution is leaky if the proof of the program does not invalidate any of the  leaky assertions. 

\begin{defn}[\textbf{Leaky Execution}]
\label{defn:leaky_execution}
Given a proof of $\{P\}S\{Q\}$ (we ignore global invariant as given in \cite{owicki1976axiomatic}),
where $S$ is an assignment or non-nested \texttt{await} statement,
and a program $T$ with at least one output statement with leaky assertion $A$ as it's pre-condition, concurrent execution of $S$ and $T$ is said to be leaky if the following two conditions hold:
\begin{enumerate}
    \item $\{Q\wedge A\}T\{Q\}$
    \item $\{P\wedge A\}T\{P\}$
\end{enumerate}
\end{defn}


\begin{prop}
Let $\{P\}S\{Q\}$ be given  and let  $T$  be a statement with pre-condition $A$, where $P$, $Q$ and $A$ are  security assertions. If [$S\parallel T$] is  leaky 
then there exists at least one leaky assertion that determinizes information in at least one secret variable.
\end{prop}

\begin{thm}[\textbf{Soundness}]
Given the proof  of $\{P\}S\{Q\}$ that has a leaky assertion $A_i$ at control point $\ell_i$ within program $S$, then there exists an execution path that leads to  truth of  $A_i$.
\end{thm}
\begin{proof}
Follows straight from the proof of $\{P\}S\{Q\}$.
\end{proof}
Issues of completeness also follow on the same lines.\newline

\noindent{\bf Application to Flow Security:\\} Dynamic labelling (DL) algorithms \cite{secrypt18} are in use for compile-time flow security certification. For a program performing IO, however, the DL algorithm could also be effectively used for identifying program points for arriving at leaky assertions at plausible control points. DL algorithm binds subjects and objects of a program with the given static or dynamic labels from a security lattice \cite{denning1976lattice} and evaluates the program as per the flow policy that enforces information flow only in upward direction of the lattice. For standard output, a static label \textit{low} is assigned denoting least confidential policy. The output statements controlled by sensitive variables that are labelled with higher security labels would cause an information flow from higher label to low, violating flow policy; thus, the algorithm easily highlights public statements with possible sensitivities.  Now, one could visualize the interactions of the various concurrent components that could possibly interfere in the functional sense. Using such an intuition, one could come up with different leaky assertions that would detect information leaks.

The process of finding information leak in programs that perform IO could be automated through the following steps: (i) first use DL algorithm to identify the public statements that could leak information; (ii) capture time-sensitivity for deriving leaky assertions corresponding to those public statements using various analysis tools or oracles; (iii) arrive at possible interaction of components at these public locations, and (iv) apply model checkers to check for certain properties in components assuming possible assertions in other components. We are working on arriving such algorithmic techniques for detecting information leaks.

%% file: related_work.tex
\section{Related Work}
\label{sec:related_work}
In this section, we briefly highlight some of the important approaches in the literature to prevent timing leaks in concurrent settings. In particular, we focus on the solution that prevents leaking information through publicly observable control points.

Our idea is primarily motivated by the work shown in \cite{owicki1976axiomatic,amtoft2004information,beringer2007secure}. The first one presented flow proof rules that use correctness properties given in \cite{owicki1976axiomatic}.  The last two approaches have given a semantic interpretation of non-interference using Hoare-like logic. However, the goal of the above approaches is to establish flow security of programs, whereas we aim to find information leaks due to IO statements and time-sensitive operations using simple inference.

Le Guernic \cite{le2007automaton} first envisages the potential danger in the use of output statements in concurrent programs, and initial solution that motivated us to provide an alternative approach in the form or security assertions. In this approach, the author develops a run-time monitor that works in tandem with a security automaton to decide whether to allow, deny or modify executions of output statements and synchronization commands, particularly those appear in a branch conditioned on sensitive variables. The automaton forbids executions of output statements in this situation. Further, it acquires locks of synchronization command before evaluating such a conditional. For instance, acquiring semaphore $sem$ in the program shown in Figure \ref{fig:observable_leak} would prevent from outputting ``acdb''. The approach provides a thoughtful solution; nonetheless, it is too restrictive. E.g., in the absence of an output statement, say \texttt{print(`c')}, the program is harmless, but as per the above solution executing the program would still be dangerous. The axiomatic approach proposed in this paper removes such restrictions on the program. Our approach enables us to assert flow insecurity in the program methodically through the ``Gries-Owicki'' conditions for functional correctness, thus overcomes the cases of false alarms.


Other than the above approach, the majority of the mechanisms found in the literature to prevent timing leaks in concurrent programs are based on either (i) developing a sound typing rules to enforce secure-flow properties in concurrent programming languages \cite{smith2001new,smith1998secure}; (ii) performing source transformation to balance the time variation or eliminate the timing channel \cite{russo2006closing,agat2000transforming,russo2006security}; or (iii) developing a run-time monitor to identify potentially dangerous executions that could leak information \cite{vassena2019foundations,stefan2012addressing,askarov2015hybrid}.


%% file: conclusions.tex
\section{Conclusions and Future Work}
\label{sec:conclusions}
In this paper, we have proposed an axiomatic proof system using leaky assertions placed at output statement locations. To the best of our knowledge,  we are the first to propose an axiomatic approach using approaches for establishing correctness of real-time distributed programs to analyze information leaks in concurrent programs that have observable public locations, leading to the following advantages: (i) it enables to formally prove the presence of information leaks; (ii) aids in integrating the notions of non-interference used in functional and security contexts with respect to complete information flow models like that of Dorothy Denning \cite{denning1976lattice}, for shared variable programs.

We are working towards automating the process of detecting information leaks in concurrent programs that exhibit execution progress through public statements. Automation has potential use in flow security certification of concurrent programs to identify information leaks via sensitive public statements. Initial work shows that DL algorithm can be integrated with model checkers for validating/speculating leaky assertions.